\documentclass[conference]{IEEEtran}
\IEEEoverridecommandlockouts
\usepackage{cite}
\usepackage{amsmath,amssymb,amsfonts}
\usepackage{algorithmic}
\usepackage{algorithm}
\usepackage{graphicx}
\usepackage[caption=false,font=footnotesize]{subfig}
\usepackage{textcomp}
\usepackage[dvipsnames,svgnames,table]{xcolor}
\usepackage{booktabs}
\usepackage{makecell}
\newcommand{\methodname}{SWEEP}
\newcommand{\method}{\textbf{\texttt{\methodname}}}
\newcommand{\gray}[1]{\textcolor{gray}{#1}}
\newcommand{\venue}[1]{{\scriptsize\textcolor{gray!70!black}{[#1]}}}
\colorlet{bestcell}{red!13}
\colorlet{secondcell}{blue!10}
\newcommand{\best}[1]{\cellcolor{bestcell}\textbf{#1}}
\newcommand{\second}[1]{\cellcolor{secondcell}\textbf{#1}}
\DeclareRobustCommand{\swatch}[1]{{\setlength{\fboxsep}{0pt}\setlength{\fboxrule}{0.4pt}\raisebox{-0.2ex}{\fcolorbox{gray!80}{#1}{\rule{0pt}{1.8ex}\hspace{5ex}}}}}
\newcommand{\band}[1]{\gray{\textit{#1}}}
\definecolor{hdrA}{HTML}{8A1F3D}
\definecolor{hdrB}{HTML}{16479B}
\definecolor{hdrC}{HTML}{15683A}
\definecolor{hdrD}{HTML}{4D4D4D}
\newcommand{\hd}[2]{\cellcolor{#1}\textcolor{white}{\textbf{#2}}}
\definecolor{oursname}{HTML}{FBE5D0}

\newcommand{\std}[1]{{\scriptsize\textcolor{ForestGreen}{($\pm$#1)}}}
\usepackage{tcolorbox}
\usepackage{ntheorem}
\theoremseparator{:}

\newtheorem{proposition}{Proposition}
\newtheorem{theorem}{Theorem}

\theoremstyle{nonumberplain}
\theorembodyfont{\normalfont}
\newtheorem{proof}{Proof}
\usepackage{hyperref}
\usepackage{balance}

\def\BibTeX{{\rm B\kern-.05em{\sc i\kern-.025em b}\kern-.08em
    T\kern-.1667em\lower.7ex\hbox{E}\kern-.125emX}}
\begin{document}

\title{Structuring and Merging Deletion Demands on Crowdsourced Channel Knowledge Maps}

\author{
\IEEEauthorblockN{Zihao Ding and Jun Huang}
\IEEEauthorblockA{Department of Electrical Engineering and Computer Science, South Dakota State University, Brookings, SD 57007, USA\\
Email: zihao.ding@jacks.sdstate.edu, jun.huang@sdstate.edu}
}

\maketitle

\begin{abstract}
A base station can build a channel knowledge map (CKM) from the uplink pilots of its terminals and sweep only the beams that the map ranks highest. A fault at a terminal or at the base station corrupts the measurements taken while it lasts, and the map must then forget them. Removing the whole terminal discards its good measurements, and the raw reports needed to remove samples are gone. We define the structured deletion unit (SDU), which names the terminals, rounds, grids, and bands to forget. We propose \underline{S}tructured \underline{W}iping with \underline{E}nvelope-bounded \underline{E}xecution \underline{P}artitioning (\method{}), which executes an SDU by clearing blocks of an aggregated ledger and re-sounding the affected grids. \method{} merges pending demands so that one campaign serves several, and it bounds what the smallest SDU containing a group clears beyond the demands. We prove that a merged execution leaves no named measurement in the map, that demands differing only in their intervals are best merged in contiguous runs, and that the greedy merge is within a factor $1+\ln Q$ of the best grouping over its candidates.
In simulation and on a USRP testbed, \method{} clears at most 5.1\% of the good measurements, against up to 86.8\% when whole terminals are removed or the map is rolled back to an earlier block.
\end{abstract}

\begin{IEEEkeywords}
Channel knowledge map, beam management, machine unlearning, software-defined radio.
\end{IEEEkeywords}

\section{Introduction}
\label{sec:intro}

Base stations keep learning from what their terminals report after deployment~\cite{Wu2026TNSE}. A channel knowledge map (CKM) stores, for each grid of the service area and each band, the average gain of every receive beam, and the base station (BS) builds it from the uplink pilots that terminals send as they move~\cite{Zeng2024COMST}. When a terminal needs a beam, the BS sweeps only the few beams that the map ranks highest for its grid, which cuts the pilot overhead of beam selection~\cite{Wu2021ICCW}. Raw reports would reveal where each terminal went~\cite{Hu2019INFOCOM}, so the BS keeps only per-grid statistics aggregated over terminals.

Faults at either end of the link corrupt such a map. At the transmitter, a terminal whose power amplifier is off calibration reports gains that are too high, and a terminal whose positioning drifts files its gains under the wrong grids~\cite{Ye2017LCN}. At the receiver, the phase calibration of the BS receive chains can drift on one band and distort the gains of every terminal on it. Such faults are found after the fact, and the map must then forget what was measured while they lasted. Existing unlearning removes a client or a sample~\cite{Romandini2025TNNLS,Liu2021IWQoS,Bourtoule2021SP}, but removing the terminal discards its good measurements, and the raw reports needed to remove samples are gone. What matches a fault is a structured deletion unit (SDU), which names the terminals, the interval of rounds, the grids, and the bands to forget.

Forgetting also costs air time. To execute an SDU, the BS clears the ledger blocks it covers and runs a re-sounding campaign, in which terminals in the affected grids sweep the full codebook. Each campaign has a fixed overhead, so executing accumulated demands one by one repeats it over the same grids. Merging them lets one campaign serve several demands, but SDUs are not closed under union. The smallest SDU that contains two demands also covers combinations that neither named, such as the rounds between their intervals, and clearing it discards good measurements, which we call over-deletion. CKM updating forgets outdated entries after an environment change as one event~\cite{Gong2025ICCCW,Han2025ICCCW}, batched unlearning executes the exact union of sample requests~\cite{Bourtoule2021SP,Hu2024CCS}, and rollback-based federated unlearning removes everything after a given round~\cite{Tao2024VLDB,Jiang2025TIFS}. None of them bounds how far a merged deletion may overreach.

The structure of the demands, however, makes such a bound tractable. A transmitter fault is narrow and long, since it names one terminal across many grids, while a receiver fault is wide and thin, since it names all terminals on one band. Where the shapes nest, merging is nearly free, and where they do not, the smallest SDU that contains both grows quickly. The BS can also price every candidate merge from the aggregated counts before running it. Merging therefore becomes a covering problem in which every group respects a budget on the good measurements it clears.

Motivated by these insights, we design Structured Wiping with Envelope-bounded Execution Partitioning (\method{}). \method{} executes an SDU as one campaign of clearing and re-sounding on a ledger aggregated over terminals. It executes each group of pending demands as its envelope, the smallest SDU that contains its members, under a bound on the utility that the envelope removes beyond the demands. It forms the groups by dynamic programming when demands differ only in their intervals and greedily otherwise, and it flushes the queue when its liability reaches a threshold.

Our contributions are summarized as follows:
\begin{itemize}
\item We formulate the deletion of faulty measurements from a crowdsourced CKM with SDUs and pose merging as a covering problem that minimizes the campaign cost in pilot slots under a per-group bound on over-deletion.
\item We propose \method{}, which executes SDUs on an aggregated ledger and merges pending demands into campaigns whose over-deletion stays within the bound.
\item We prove that merging leaves no named measurement, that contiguous groups are optimal when demands differ only in intervals, and that the greedy merge is within $1+\ln Q$ of the best cover over its candidate groups.
\item In simulation and on a USRP testbed, \method{} clears at most 2.6\% and 5.1\% of the good measurements, against up to 86.8\% for terminal removal or rollback.
\end{itemize}

\section{System Model and Problem Formulation}
\label{sec:system}

\begin{figure}[t]
  \centering
  \includegraphics[width=0.95\columnwidth]{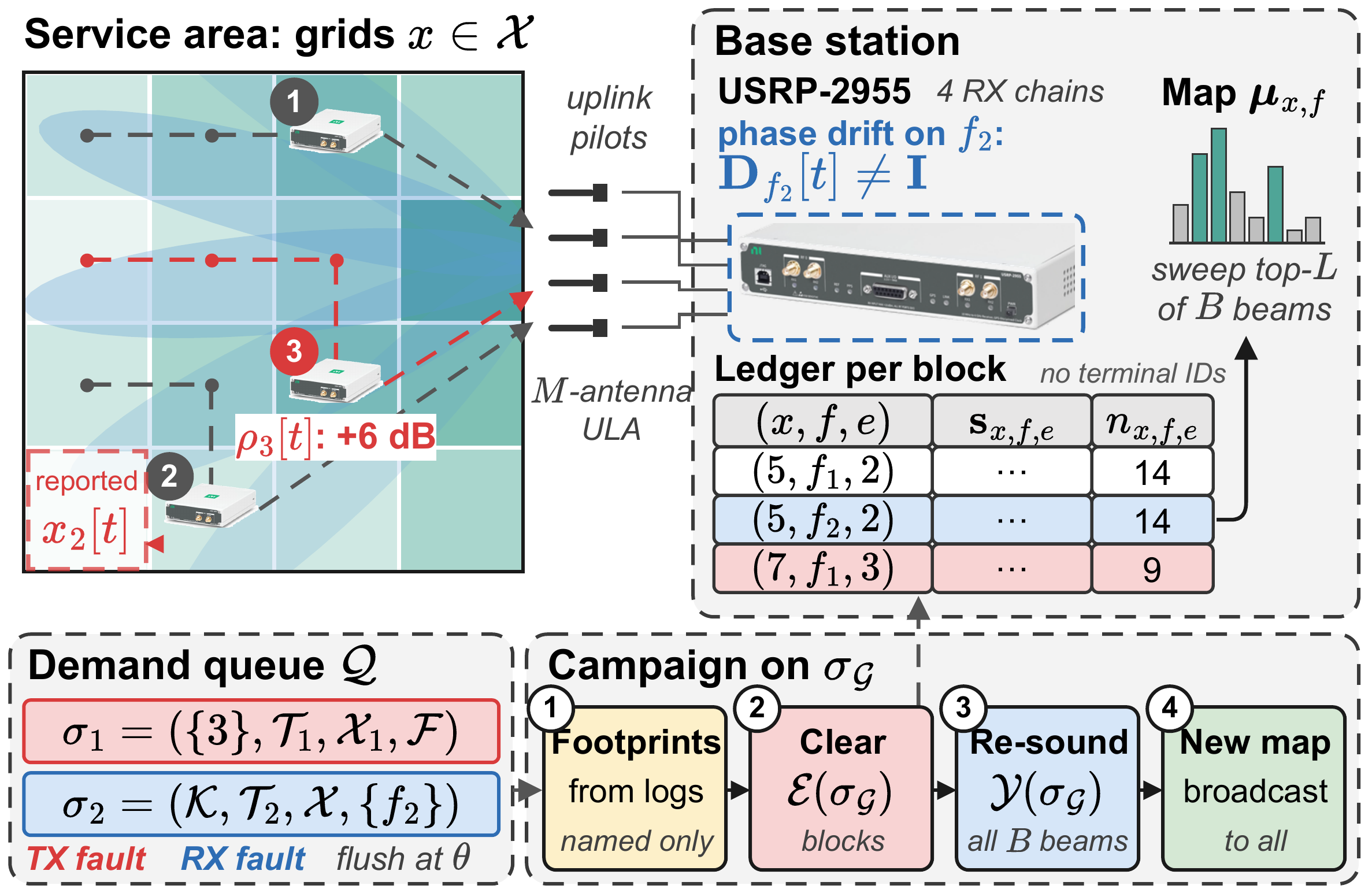}
  \caption{System model of structured deletion on a channel knowledge map.}
  \label{fig:system}
\end{figure}

\begin{figure*}[t]
  \centering
  \begin{minipage}[t]{0.48\textwidth}
    \centering
    \includegraphics[width=\linewidth]{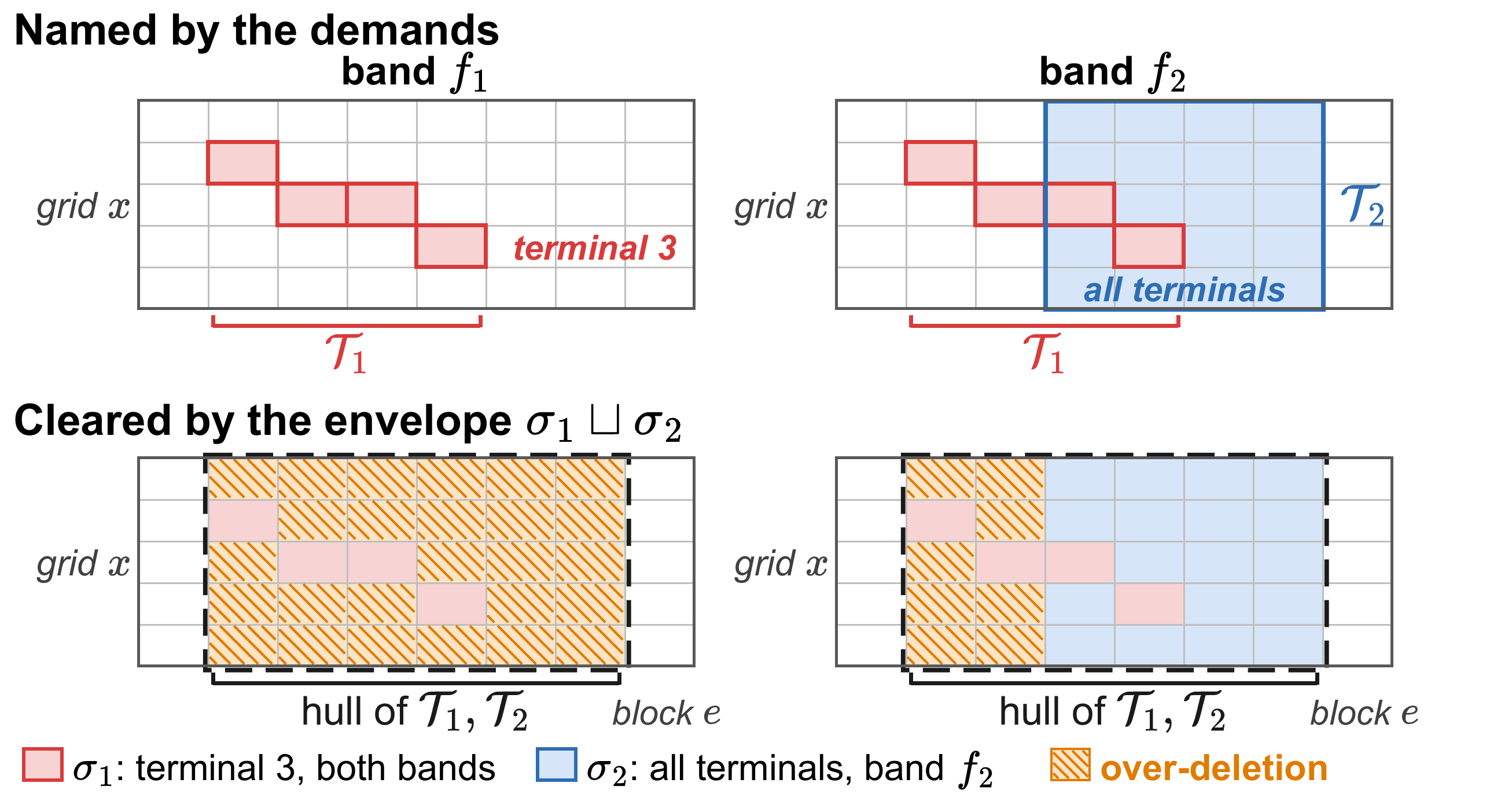}
    \caption{Two fault shapes and the blocks that their envelope clears.}
    \label{fig:shape}
  \end{minipage}
  \hfill
  \begin{minipage}[t]{0.48\textwidth}
    \centering
    \includegraphics[width=\linewidth]{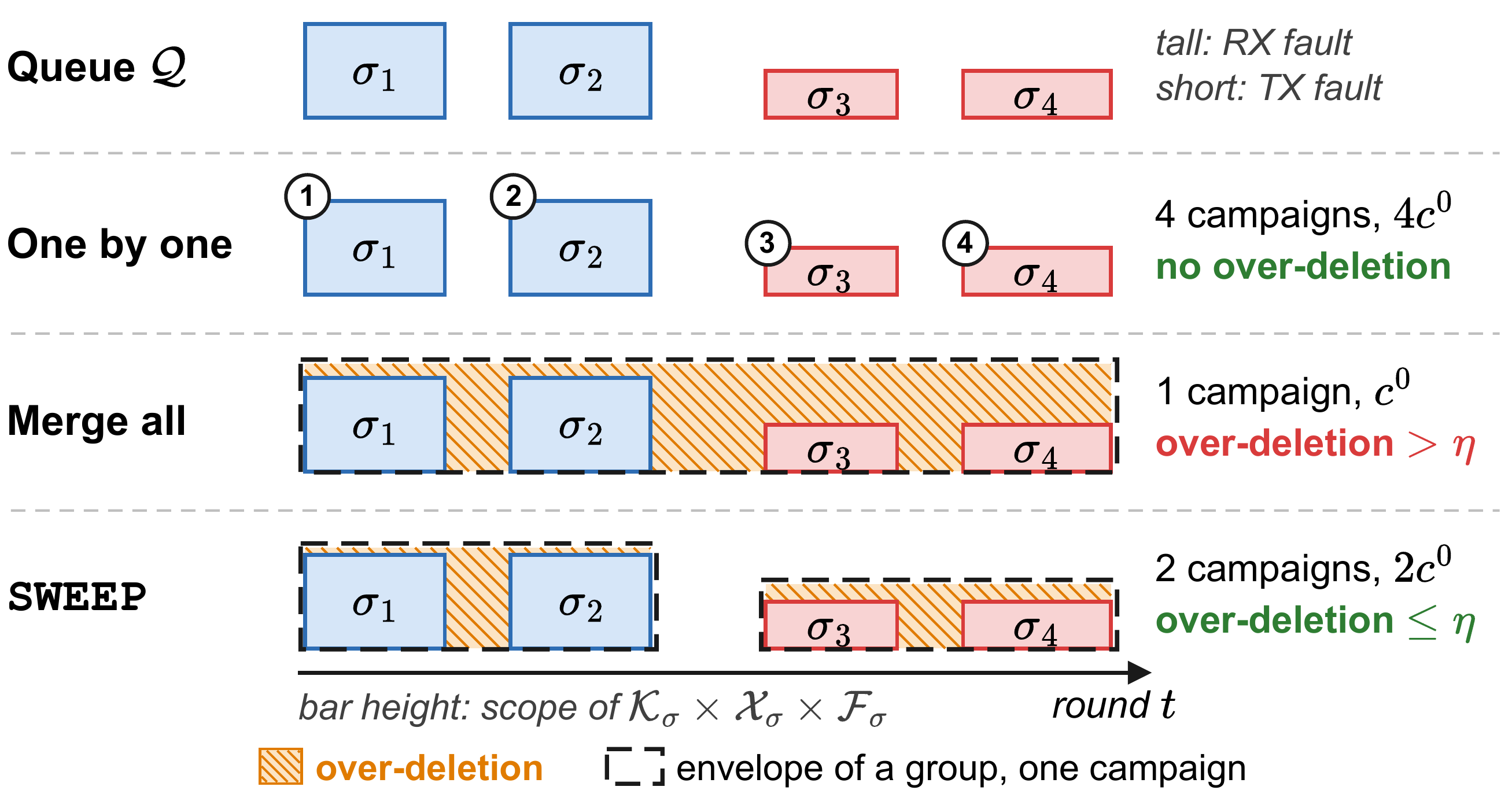}
    \caption{Four demands merged one by one, all at once, and by \method{}.}
    \label{fig:merge}
  \end{minipage}
\end{figure*}

\subsection{System Model}
\label{subsec:system-model}

We consider a base station (BS) that serves a set $\mathcal{K}$ of $K$ terminals on a set $\mathcal{F}$ of bands, as illustrated in Fig.~\ref{fig:system}. The BS has an $M$-antenna uniform linear array and a codebook of $B$ receive beams $\mathbf{v}_1,\ldots,\mathbf{v}_B\in\mathbb{C}^{M}$, and the service area is divided into a set $\mathcal{X}$ of grids. In every round $t$, terminal $k$ reports its grid $x_k[t]\in\mathcal{X}$ and sends an uplink pilot on every band. On band $f$, the BS measures the gain vector $\boldsymbol{\gamma}_{k,f}[t]\in\mathbb{R}^{B}$ with entries
\begin{equation}
    \gamma_{k,f,b}[t]\;=\;\rho_k[t]\,\bigl|\mathbf{v}_b^{\mathrm{H}}\,\mathbf{D}_f[t]\,\mathbf{h}_{k,f}[t]\bigr|^{2},\quad b=1,\ldots,B,
    \label{eq:gain}
\end{equation}
where $\mathbf{h}_{k,f}[t]\in\mathbb{C}^{M}$ is the uplink channel scaled by the transmit power that the terminal declares. The ratio $\rho_k[t]$ of the actual to the declared power is one for a calibrated terminal, and the diagonal matrix $\mathbf{D}_f[t]$ of residual phases of the receive chains on band $f$ is the identity for a calibrated BS.

The BS keeps no per-terminal records, which would be location traces. It groups the rounds into blocks of $W$ rounds, indexed by $e$, and keeps a ledger with one entry per grid $x$, band $f$, and block $e$. The entry holds the sum $\mathbf{s}_{x,f,e}$ of the gain vectors reported from grid $x$ on band $f$ during block $e$ and their number $n_{x,f,e}$. Each terminal keeps its own reports and diagnostics in a local log. The map is
\begin{equation}
    \boldsymbol{\mu}_{x,f}\;=\;\frac{\sum_{e}\mathbf{s}_{x,f,e}}{\sum_{e}n_{x,f,e}},
    \label{eq:map}
\end{equation}
and grid $x$ is valid on band $f$ once $\sum_{e}n_{x,f,e}\ge n^{\textrm{min}}$.

When a terminal in grid $x$ needs a beam on band $f$, the BS sweeps only the set $\mathcal{L}_{x,f}$ of the $L\ll B$ beams with the largest entries of $\boldsymbol{\mu}_{x,f}$, or all $B$ beams if the grid is not valid. The map thus saves $B-L$ pilot slots per selection at the price of a signal-to-noise ratio (SNR) loss
\begin{equation}
    \phi_{x,f}\;=\;\mathbb{E}\biggl[10\log_{10}\frac{\max_{b}\gamma_{b}}{\max_{b\in\mathcal{L}_{x,f}}\gamma_{b}}\biggr]
    \label{eq:loss}
\end{equation}
in dB, where $\gamma_b$ is the gain of beam $b$ for that terminal. Let $\lambda_{x,f}$ denote the traffic of grid $x$ on band $f$, the number of beam selections it requests per round.

\subsection{Problem Formulation}
\label{subsec:problem-formulation}

A deletion demand arrives when a diagnostic log reveals a fault and the rounds in which it lasted. We take the demands as given and do not attribute faults to their sources. A demand is a structured deletion unit (SDU)
\begin{equation}
    \sigma\;=\;\bigl(\mathcal{K}_\sigma,\;\mathcal{T}_\sigma,\;\mathcal{X}_\sigma,\;\mathcal{F}_\sigma\bigr),
    \label{eq:sdu}
\end{equation}
which names terminals $\mathcal{K}_\sigma\subseteq\mathcal{K}$, an interval of rounds $\mathcal{T}_\sigma=[t_\sigma^{\textrm{s}},t_\sigma^{\textrm{e}}]$, grids $\mathcal{X}_\sigma\subseteq\mathcal{X}$, and bands $\mathcal{F}_\sigma\subseteq\mathcal{F}$. It asks the BS to forget every measurement that a terminal in $\mathcal{K}_\sigma$ reported from a grid in $\mathcal{X}_\sigma$ on a band in $\mathcal{F}_\sigma$ during $\mathcal{T}_\sigma$. The faults in Eq.~\eqref{eq:gain} yield three shapes. A power fault of terminal $k$ gives $(\{k\},\mathcal{T}_\sigma,\mathcal{X},\mathcal{F})$, and a positioning fault gives $(\{k\},\mathcal{T}_\sigma,\mathcal{X}_\sigma,\mathcal{F})$, where $\mathcal{X}_\sigma$ holds the grids that $k$ reported under a poor position fix. A receive-chain fault on band $f$ gives $(\mathcal{K},\mathcal{T}_\sigma,\mathcal{X},\{f\})$. Fig.~\ref{fig:shape} draws a transmitter fault and a receiver fault on the ledger.

The ledger cannot separate terminals within a block, so the BS executes an SDU by clearing every block that holds a named measurement. For $k\in\mathcal{K}_\sigma$, $t\in\mathcal{T}_\sigma$, and $f\in\mathcal{F}_\sigma$, these blocks are
\begin{equation}
    \mathcal{E}(\sigma)\;=\;\bigl\{(x_k[t],\,f,\,\lceil t/W\rceil):\,x_k[t]\in\mathcal{X}_\sigma\bigr\}.
    \label{eq:clear}
\end{equation}
The BS then runs a campaign in which the terminals in every grid--band pair of $\mathcal{Y}(\sigma)=\{(x,f):(x,f,e)\in\mathcal{E}(\sigma)\}$ sweep all $B$ beams until $n^{\textrm{min}}$ new measurements are collected there, and it broadcasts the new map. The campaign costs
\begin{equation}
    C(\sigma)\;=\;c^{0}+n^{\textrm{min}}B\,|\mathcal{Y}(\sigma)|
    \label{eq:cost}
\end{equation}
pilot slots, where $c^{0}$ counts the announcement, the synchronization, and the map broadcast in the same unit. Clearing also discards the good measurements in the cleared blocks. The utility that a set $\mathcal{E}$ of blocks removes is the traffic-weighted number of measurements it holds, $\Delta(\mathcal{E})=\sum_{(x,f,e)\in\mathcal{E}}\lambda_{x,f}\,n_{x,f,e}$, and we write $\Delta(\sigma)$ for $\Delta(\mathcal{E}(\sigma))$.

Demands accumulate in a queue $\mathcal{Q}$ of $Q$ SDUs $\sigma_1,\ldots,\sigma_Q$, whose fields we abbreviate as $\mathcal{K}_i$, $\mathcal{T}_i$, $\mathcal{X}_i$, and $\mathcal{F}_i$. Executing them one by one pays $c^{0}$ per demand and re-sounds the same pairs repeatedly. A group $\mathcal{G}\subseteq\mathcal{Q}$ can share one campaign, but SDUs are closed under intersection and not under union. The smallest SDU that contains every member of $\mathcal{G}$ is its envelope
\begin{equation}
    \sigma_{\mathcal{G}}=\Bigl(\bigcup_{i\in\mathcal{G}}\mathcal{K}_i,\;\bigl[\min_{i\in\mathcal{G}}t_i^{\textrm{s}},\,\max_{i\in\mathcal{G}}t_i^{\textrm{e}}\bigr],\;\bigcup_{i\in\mathcal{G}}\mathcal{X}_i,\;\bigcup_{i\in\mathcal{G}}\mathcal{F}_i\Bigr).
    \label{eq:envelope}
\end{equation}
The envelope clears every block that its members name, but also blocks that none of them named, such as those between two intervals or of one terminal in the grids of another. The utility of these extra blocks is the over-deletion of the group. The BS selects a collection $\mathcal{P}$ of groups that together contain every demand by solving
\begin{equation}
\begin{split}
    (\mathrm{P}):\;\;\min_{\mathcal{P}}\;\;
    & \sum_{\mathcal{G}\in\mathcal{P}} C(\sigma_{\mathcal{G}}) \\
    \mathrm{s.t.}\;\;
    & (\mathrm{C}_1):\;\Delta(\sigma_{\mathcal{G}})\;\le\;\Delta\bigl({\textstyle\bigcup_{i\in\mathcal{G}}}\mathcal{E}(\sigma_i)\bigr)+\eta,
\end{split}
\label{eq:problem}
\end{equation}
where $\eta$ bounds the over-deletion of every group $\mathcal{G}\in\mathcal{P}$. A cleared block stays cleared, so groups may share a demand, and $\mathcal{P}$ is a cover of $\mathcal{Q}$ rather than a partition.

\section{\texorpdfstring{\method{}}{SWEEP}: Structured Wiping with Envelope-bounded Execution Partitioning}
\label{sec:method}

Problem~$(\mathrm{P})$ in Eq.~\eqref{eq:problem} cannot be solved directly. The BS cannot price a group from the ledger alone, since Eq.~\eqref{eq:clear} depends on where the named terminals reported. The problem ranges over the $2^{Q}$ subsets of the queue, and meanwhile the map keeps serving traffic with the measurements that the queue names. To address these issues, we design \underline{\textbf{S}}tructured \underline{\textbf{W}}iping with \underline{\textbf{E}}nvelope-bounded \underline{\textbf{E}}xecution \underline{\textbf{P}}artitioning (\method{}), whose three modules follow in Sections~\ref{subsec:ledger}, \ref{subsec:merge}, and~\ref{subsec:flush}. Fig.~\ref{fig:arch} shows when each runs.

\subsection{Block Ledger and Footprints}
\label{subsec:ledger}

This module keeps the map and executes the campaign of one SDU. In every round, the BS adds each report to the ledger entry of its grid, band, and block, updates Eq.~\eqref{eq:map} for the touched grids, and drops the report. The ledger holds $|\mathcal{X}||\mathcal{F}|(B+1)$ numbers per block regardless of $K$. The block length $W$ sets the time resolution of deletion. An interval that does not align with the block boundaries clears its partial blocks whole, so a small $W$ clears less but stores more.

To price and execute SDUs, the BS needs to know where the named terminals reported, which the ledger does not record. When the queue is flushed, every terminal named in a demand sends its footprint, the set of blocks in which it reported during the rounds that the queue spans, read from its local log. A demand that names all terminals needs no footprint, since Eq.~\eqref{eq:clear} then selects every nonempty block of its grids, bands, and interval. From the footprints and the counts $n_{x,f,e}$, the BS evaluates $\mathcal{E}(\sigma)$, $C(\sigma)$, and $\Delta(\sigma)$ for any SDU built from the queue without further signaling. A footprint lists at most one grid per round, so it takes $\lceil\log_2|\mathcal{X}|\rceil$ bits per round that the queue spans, a few hundred bytes for hundreds of rounds. The BS deletes the footprints after the flush, and a terminal that is never named sends nothing, so the BS holds no trace of it.

To execute an SDU $\sigma$, the BS clears the entries of $\mathcal{E}(\sigma)$, including the good measurements that other terminals added to them, which is the price of keeping no per-terminal records and is counted by $\Delta$. It marks every pair of $\mathcal{Y}(\sigma)$ whose count falls below $n^{\textrm{min}}$ as not valid, so that its selections fall back to the full sweep. It then schedules the re-sounding campaign of Section~\ref{subsec:problem-formulation}. The new measurements enter blocks after the current round and hence outside every named interval. Eq.~\eqref{eq:map} is a ratio of sums over blocks, so removing an entry removes its sum and its count together, and the map is exactly the one that a ledger without these blocks would give.

\subsection{Envelope-Bounded Merge}
\label{subsec:merge}

This module solves $(\mathrm{P})$ on the queue at a flush. It first builds the candidate groups. For every pair $i,j$ of demands, the candidate $\mathcal{G}_{ij}$ holds every demand whose SDU lies inside the envelope of $\sigma_i$ and $\sigma_j$, and the candidates that violate $(\mathrm{C}_1)$ are dropped. A demand inside another thus joins its group at no cost, since the envelope in Eq.~\eqref{eq:envelope} does not change. Let $\mathcal{A}$ denote the remaining candidates. The candidate $\mathcal{G}_{ii}$ has the envelope $\sigma_i$ and satisfies $(\mathrm{C}_1)$, so $\mathcal{A}$ covers the queue. The interval of any envelope is fixed by two members, and $\mathcal{A}$ holds every group whose other fields are also spanned by two of its members. A group whose envelope needs three demands to span, such as three transmitter faults in three separate areas, crosses several trajectories and rarely meets $(\mathrm{C}_1)$ under a small $\eta$.

The two shapes fare differently under this rule. A transmitter fault on the band of a receiver fault and within its interval lies inside it and joins at no cost. A transmitter fault on both bands spreads the envelope of a receiver fault over the other band of every grid, so they merge only if that band holds few measurements then. Transmitter faults of two terminals merge when their trajectories pass the same grids in nearby rounds.

When all demands share their terminals, grids, and bands, as for repeated drifts of one receive chain, \method{} sorts the demands by their start rounds and solves $(\mathrm{P})$ exactly by the dynamic program
\begin{equation}
    F(j)\;=\;\min_{1\le i\le j,\;\mathcal{G}_{ij}\in\mathcal{A}}\Bigl[F(i-1)+C\bigl(\sigma_{\mathcal{G}_{ij}}\bigr)\Bigr],\quad F(0)=0,
    \label{eq:dp}
\end{equation}
whose optimality \textbf{Proposition}~\ref{prop:interval} in Section~\ref{sec:analysis} establishes. In general, \method{} covers the queue greedily~\cite{Chvatal1979MOR}. Let $\mathcal{U}$ denote the demands not yet covered, initially $\mathcal{Q}$. The greedy repeatedly selects the candidate $\mathcal{G}\in\mathcal{A}$ with the least cost per uncovered demand, $C(\sigma_{\mathcal{G}})/|\mathcal{G}\cap\mathcal{U}|$, adds it to $\mathcal{P}$, and removes its members from $\mathcal{U}$ until $\mathcal{U}$ is empty. \textbf{Theorem}~\ref{thm:greedy} bounds its cost. Fig.~\ref{fig:merge} shows the effect on four demands. Executing them one by one pays four fixed overheads, and merging all of them clears the long gaps between their intervals. \method{} merges only the demands whose envelopes stay within $\eta$.

Both procedures use only the ledger counts and the footprints and evaluate at most $Q^{2}$ candidates or runs.

\subsection{Liability-Triggered Flush}
\label{subsec:flush}

This module decides when to flush. While a demand waits, the map keeps serving traffic with the measurements it names. \method{} therefore charges each pending demand $\sigma_i$ the liability rate $\psi_i=\sum_{(x,f)\in\mathcal{Y}(\sigma_i)}\lambda_{x,f}$ per round, the traffic that its grid--band pairs serve, and it tracks the total rate $\Psi=\sum_{i\in\mathcal{Q}}\psi_i$ of the queue. As soon as $\Psi$ reaches a threshold $\theta$, it collects the footprints, runs the merge, executes the campaigns of the selected groups, and empties the queue. A small $\theta$ flushes often and merges little, while a large $\theta$ merges more but lets the named measurements serve longer. A receiver fault names every pair of its band and thus triggers a flush almost at once. Demands that arrive during a flush wait for the next one. \textbf{Proposition}~\ref{prop:flush} shows that a threshold of this form is optimal when the cost of a flush does not depend on the queue.

\method{} transmits only the footprints and the campaigns of the selected groups, and its merge is polynomial in $Q$.

\begin{figure*}[t]
  \centering
  \begin{minipage}[t]{0.48\textwidth}
    \centering
    \includegraphics[width=\linewidth]{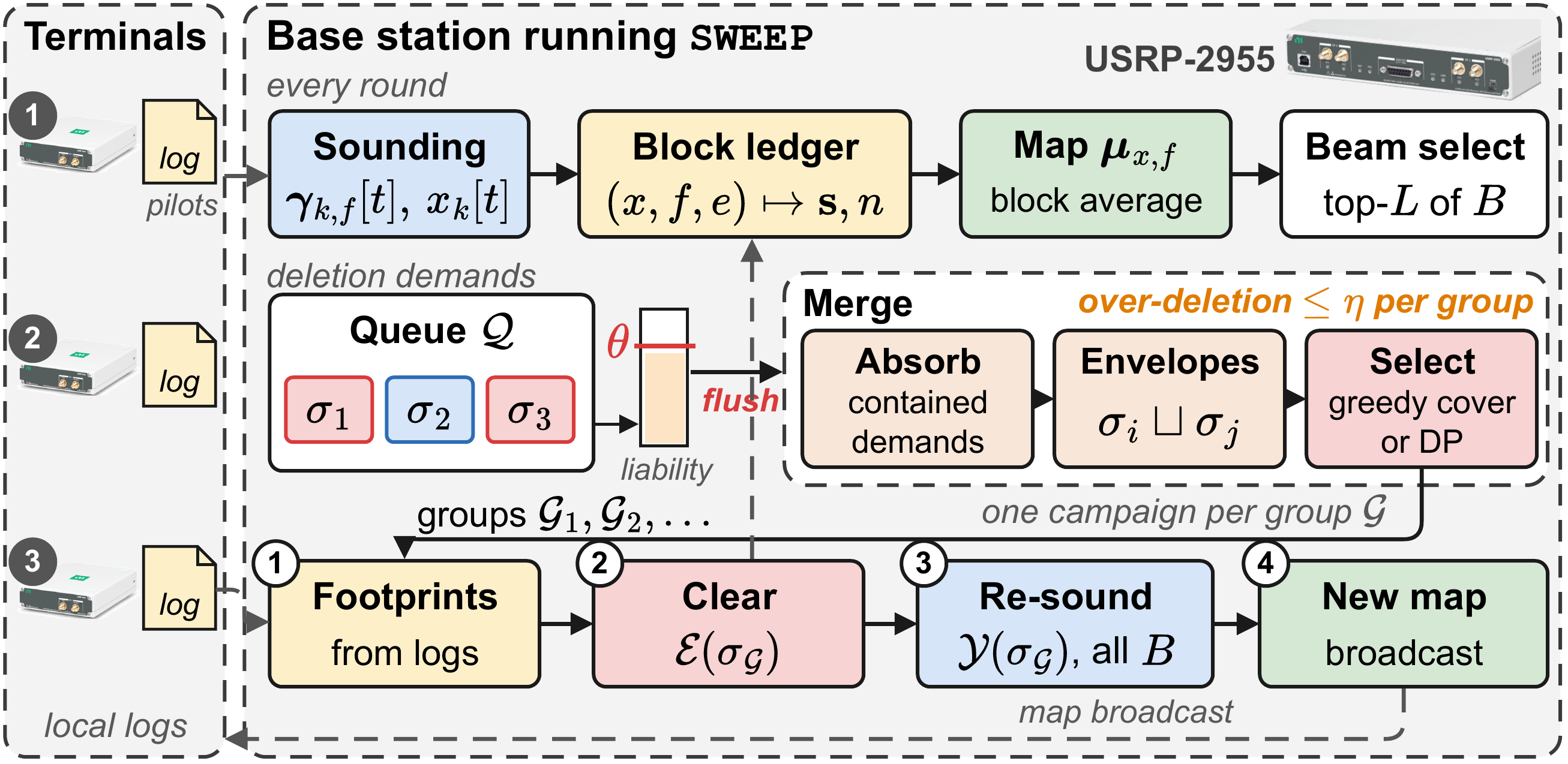}
    \caption{The architecture of \method{} at the base station.}
    \label{fig:arch}
  \end{minipage}
  \hfill
  \begin{minipage}[t]{0.48\textwidth}
    \centering
    \includegraphics[width=\linewidth]{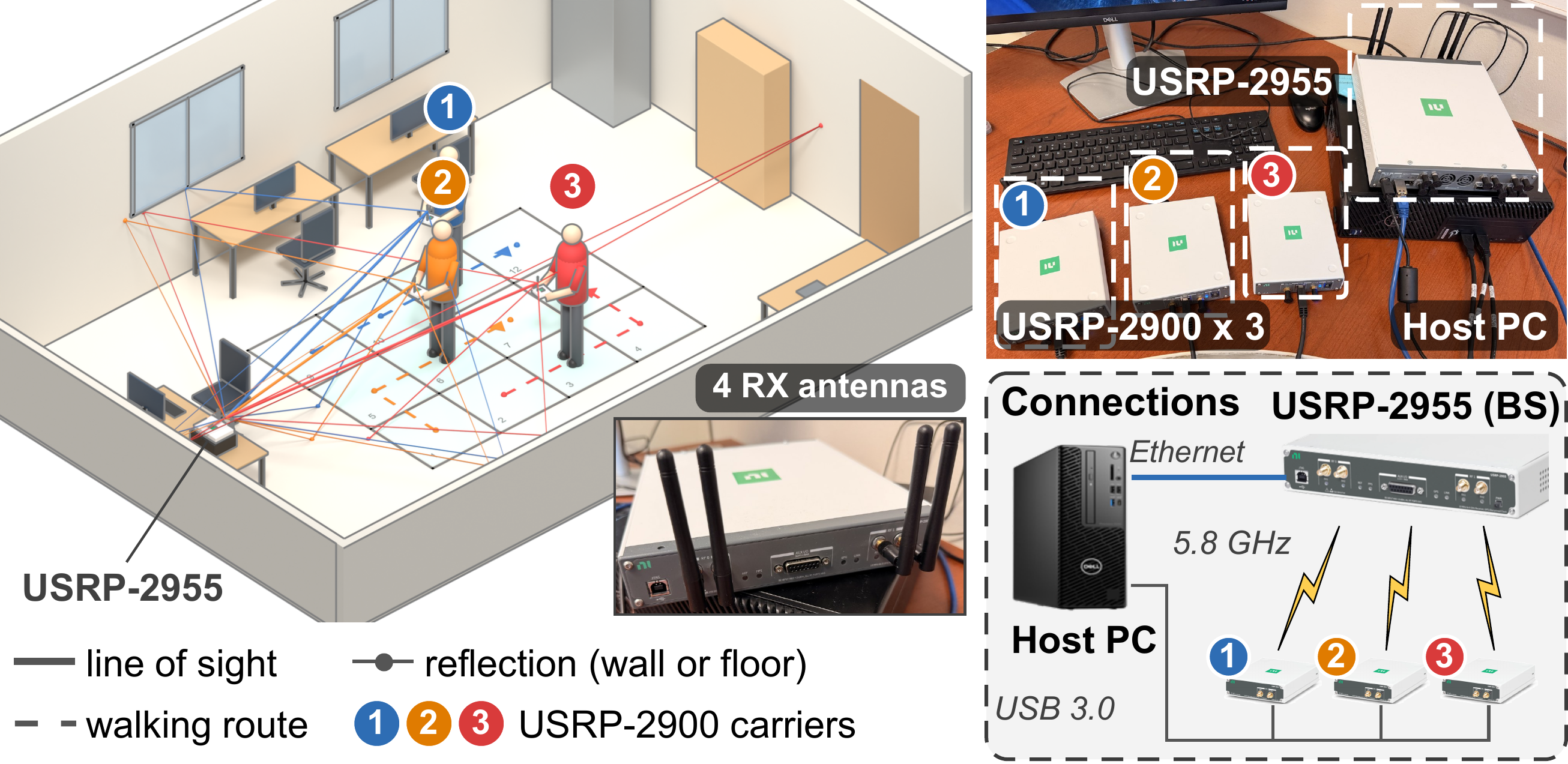}
    \caption{USRP testbed: lab scene with traced rays, devices, and links.}
    \label{fig:testbed}
  \end{minipage}
\end{figure*}

\begin{table*}[!t]
\centering
\setlength{\tabcolsep}{4pt}
\setlength{\aboverulesep}{0pt}
\setlength{\belowrulesep}{0pt}
\setlength{\extrarowheight}{0pt}
\renewcommand{\arraystretch}{0.88}
\caption{Map quality, deletion, and execution cost per method and fault type. The SNR loss and $\mathrm{GainErr}$ are in dB, $\mathrm{Hit}$ and $\mathrm{Invalid}$ in \% of the selections and of the traffic, $\mathrm{Left}$ and $\mathrm{Over}$ in \% of the named and of the good measurements, $\mathrm{Wait}$ in rounds, the pilot slots in thousands, and the liability in traffic-rounds. The testbed columns come from the USRP devices, with the same units and the air time of re-sounding in seconds. A \swatch{bestcell} cell marks the best value and a \swatch{secondcell} cell the second best within each fault type, excluding Retrain.}
\label{tab:main}
\resizebox{\linewidth}{!}{%
\begin{tabular}{l|cccc|ccc|ccc|cccc}
\Xhline{1.2pt}
 & \multicolumn{4}{c|}{\hd{hdrA}{Map}} & \multicolumn{3}{c|}{\hd{hdrB}{Deletion}} & \multicolumn{3}{c|}{\hd{hdrC}{Execution}} & \multicolumn{4}{c}{\hd{hdrD}{Testbed}}\\
\textbf{Method} & SNR loss$\downarrow$ & Hit$\uparrow$ & GainErr$\downarrow$ & Invalid$\downarrow$ & Left$\downarrow$ & Over$\downarrow$ & Wait$\downarrow$ & Campaigns$\downarrow$ & Pilots$\downarrow$ & Liability$\downarrow$ & SNR loss$\downarrow$ & Hit$\uparrow$ & Over$\downarrow$ & Air time$\downarrow$\\
\Xhline{1.2pt}
\multicolumn{15}{l}{\band{(a) Transmitter faults}}\\
\Xhline{0.6pt}
SISA \venue{S\&P'21} & 0.816\std{0.142} & 90.80\std{0.86} & 3.574\std{0.245} & \best{0.43}\std{0.29} & 36.5\std{7.8} & 1.8\std{0.4} & 49.1\std{10.8} & 17.0\std{1.7} & 60.7\std{3.5} & 683\std{230} & 0.918\std{0.087} & 91.26\std{0.96} & 5.7\std{1.2} & 5.616\std{0.438}\\
ERASER \venue{CCS'24} & 0.809\std{0.141} & 90.86\std{0.86} & 3.552\std{0.239} & \best{0.43}\std{0.29} & 30.9\std{7.2} & \second{1.7}\std{1.2} & 59.8\std{15.8} & \second{10.3}\std{1.8} & \second{49.9}\std{4.1} & 434\std{609} & 0.876\std{0.063} & 91.88\std{0.82} & \second{3.9}\std{1.4} & 2.680\std{0.296}\\
FedRecover \venue{S\&P'23} & 0.805\std{0.144} & \second{90.90}\std{0.88} & \second{3.550}\std{0.237} & 0.64\std{0.29} & \second{27.7}\std{5.6} & 17.2\std{2.5} & \second{46.7}\std{7.3} & 12.0 & 134.9\std{7.3} & \best{0} & 0.837\std{0.074} & 92.37\std{0.91} & 22.8\std{4.7} & 3.912\std{0.351}\\
Crab \venue{TIFS'25} & \second{0.778}\std{0.170} & 90.85\std{1.12} & \best{3.546}\std{0.232} & 6.29\std{5.65} & \best{19.0}\std{4.7} & 79.2\std{6.7} & \best{40.3}\std{3.1} & 11.7\std{0.5} & 438.6\std{83.7} & \best{0} & \second{0.796}\std{0.092} & \second{92.94}\std{1.16} & 67.1\std{6.3} & 4.032\std{0.274}\\
ST-REM \venue{ToN'20} & \best{0.755}\std{0.140} & \best{91.37}\std{0.98} & 3.826\std{0.355} & 10.93\std{8.03} & 65.8\std{15.1} & 17.3\std{9.7} & 193.9\std{25.0} & \best{0.0} & 113.2\std{83.1} & 6683\std{2829} & \best{0.774}\std{0.081} & \best{93.17}\std{0.88} & 20.6\std{5.2} & \best{0.000}\\
\Xhline{0.8pt}
\cellcolor{oursname}\textit{\textbf{\methodname{}}} & 0.808\std{0.141} & 90.87\std{0.85} & 3.551\std{0.237} & 0.44\std{0.30} & 28.4\std{5.9} & \best{1.1}\std{0.2} & 59.3\std{16.0} & 11.7\std{0.5} & \best{49.8}\std{4.0} & 5\std{7} & 0.829\std{0.058} & 92.45\std{0.73} & \best{3.2}\std{0.6} & \second{2.648}\std{0.183}\\
\Xhline{0.8pt}
\textit{Retrain} & 0.789\std{0.145} & 91.05\std{0.89} & 3.488\std{0.236} & 0.48\std{0.34} & 0.0 & 0.0 & - & - & - & - & 0.716\std{0.069} & 93.64\std{0.79} & 0.0 & -\\
\Xhline{0.8pt}
\multicolumn{15}{l}{\band{(b) Receiver faults}}\\
\Xhline{0.6pt}
SISA \venue{S\&P'21} & 0.888\std{0.168} & 89.80\std{1.54} & 3.801\std{0.279} & \best{0.67}\std{0.49} & 42.4\std{8.0} & 13.0\std{4.3} & 72.6\std{11.4} & 8.0 & \best{67.1}\std{5.2} & 6927\std{1524} & 1.067\std{0.126} & 89.79\std{1.37} & 14.8\std{3.6} & 1.888\std{0.147}\\
ERASER \venue{CCS'24} & 0.862\std{0.169} & \second{90.21}\std{1.31} & 3.729\std{0.296} & \second{0.75}\std{0.70} & 33.8\std{9.0} & \second{1.8}\std{0.4} & 67.7\std{16.2} & \second{3.9}\std{0.3} & \second{90.5}\std{11.2} & 4334\std{2537} & 0.972\std{0.094} & 90.96\std{1.09} & \second{5.6}\std{1.7} & \second{1.232}\std{0.091}\\
FedRecover \venue{S\&P'23} & 1.227\std{0.328} & 87.60\std{2.17} & 3.936\std{0.257} & 21.87\std{1.37} & \best{21.9}\std{5.0} & 86.8\std{2.6} & 50.3\std{5.5} & 4.0 & 411.1\std{16.0} & \best{0} & 1.384\std{0.218} & 86.73\std{2.04} & 79.4\std{5.8} & 1.376\std{0.063}\\
Crab \venue{TIFS'25} & \second{0.857}\std{0.196} & \second{90.21}\std{1.49} & \second{3.657}\std{0.267} & 2.31\std{2.45} & \best{21.9}\std{5.0} & 49.3\std{4.7} & \best{47.7}\std{5.1} & 4.0 & 196.2\std{27.6} & \best{0} & \second{0.948}\std{0.103} & \second{91.35}\std{1.25} & 46.7\std{7.1} & 1.352\std{0.082}\\
ST-REM \venue{ToN'20} & 1.571\std{0.510} & 85.64\std{3.36} & 5.061\std{0.618} & 43.94\std{8.79} & 26.0\std{9.4} & 72.1\std{6.2} & 207.2\std{25.0} & \best{0.0} & 455.2\std{91.1} & 15735\std{3225} & 1.623\std{0.267} & 84.69\std{2.33} & 68.3\std{8.4} & \best{0.000}\\
\Xhline{0.8pt}
\cellcolor{oursname}\textit{\textbf{\methodname{}}} & \best{0.849}\std{0.170} & \best{90.35}\std{1.31} & \best{3.649}\std{0.268} & 0.84\std{0.74} & 22.9\std{5.0} & \best{1.7}\std{0.4} & \second{47.9}\std{5.2} & 4.0 & 93.0\std{8.8} & \best{0} & \best{0.903}\std{0.111} & \best{91.86}\std{1.18} & \best{3.8}\std{0.8} & 1.280\std{0.105}\\
\Xhline{0.8pt}
\textit{Retrain} & 0.771\std{0.153} & 91.14\std{1.03} & 3.465\std{0.251} & 1.36\std{1.37} & 0.0 & 0.0 & - & - & - & - & 0.738\std{0.083} & 93.42\std{0.94} & 0.0 & -\\
\Xhline{0.8pt}
\multicolumn{15}{l}{\band{(c) Mixed faults}}\\
\Xhline{0.6pt}
SISA \venue{S\&P'21} & \second{0.935}\std{0.158} & \second{89.49}\std{1.28} & 3.843\std{0.223} & 1.20\std{0.94} & 29.6\std{7.7} & \second{23.9}\std{5.7} & 45.5\std{7.2} & 27.0\std{1.6} & 186.0\std{13.6} & 2298\std{1071} & 1.143\std{0.149} & 89.12\std{1.64} & \second{25.7}\std{5.4} & 7.504\std{0.683}\\
ERASER \venue{CCS'24} & 0.942\std{0.172} & 89.43\std{1.37} & 3.889\std{0.295} & \second{1.14}\std{0.95} & 33.7\std{9.3} & 30.4\std{18.5} & 68.1\std{14.2} & \second{11.7}\std{2.6} & \second{162.4}\std{13.3} & 4889\std{3391} & 1.172\std{0.176} & 88.84\std{1.87} & 32.9\std{9.6} & 3.912\std{0.457}\\
FedRecover \venue{S\&P'23} & 1.295\std{0.307} & 87.04\std{1.96} & 4.132\std{0.217} & 22.43\std{1.47} & \second{23.5}\std{7.6} & 81.6\std{7.1} & \second{45.1}\std{10.2} & 16.0 & 468.5\std{14.6} & \best{0} & 1.497\std{0.241} & 85.96\std{2.16} & 82.6\std{4.9} & 5.264\std{0.392}\\
Crab \venue{TIFS'25} & 0.960\std{0.180} & 89.27\std{1.32} & \second{3.828}\std{0.271} & 9.80\std{7.31} & \best{15.9}\std{5.8} & 78.7\std{3.6} & \best{38.8}\std{3.2} & 15.9\std{0.3} & 565.0\std{80.5} & \best{0} & \second{1.106}\std{0.138} & \second{89.63}\std{1.52} & 76.3\std{6.8} & 5.384\std{0.516}\\
ST-REM \venue{ToN'20} & 1.532\std{0.504} & 85.57\std{3.05} & 5.087\std{0.693} & 34.27\std{9.38} & 35.6\std{8.4} & 59.8\std{8.0} & 208.2\std{20.6} & \best{0.0} & 355.0\std{97.1} & 23981\std{3859} & 1.708\std{0.293} & 83.75\std{2.71} & 58.7\std{10.3} & \best{0.000}\\
\Xhline{0.8pt}
\cellcolor{oursname}\textit{\textbf{\methodname{}}} & \best{0.896}\std{0.162} & \best{89.91}\std{1.28} & \best{3.756}\std{0.243} & \best{0.99}\std{0.91} & 25.0\std{7.3} & \best{2.6}\std{0.3} & 59.8\std{19.1} & 15.7\std{0.5} & \best{138.4}\std{11.6} & 3\std{4} & \best{0.957}\std{0.117} & \best{91.31}\std{1.29} & \best{5.1}\std{1.1} & \second{3.336}\std{0.264}\\
\Xhline{0.8pt}
\textit{Retrain} & 0.771\std{0.154} & 91.11\std{1.05} & 3.462\std{0.254} & 1.60\std{1.65} & 0.0 & 0.0 & - & - & - & - & 0.751\std{0.096} & 93.07\std{1.03} & 0.0 & -\\
\Xhline{1.2pt}
\end{tabular}}
\end{table*}

\section{Theoretical Analysis}
\label{sec:analysis}

This section states the guarantees of Section~\ref{sec:method}.

\begin{proposition}\label{prop:complete}
For every group $\mathcal{G}$ and every $i\in\mathcal{G}$, $\mathcal{E}(\sigma_i)\subseteq\mathcal{E}(\sigma_{\mathcal{G}})$. After the campaign of $\sigma_{\mathcal{G}}$ clears its blocks, no measurement named by $\sigma_i$ remains in the ledger, and $\boldsymbol{\mu}_{x,f}$ equals the map of a ledger that never received the cleared blocks.
\end{proposition}

\begin{proof}
By Eq.~\eqref{eq:envelope}, every field of $\sigma_{\mathcal{G}}$ contains the corresponding field of $\sigma_i$, and every condition in Eq.~\eqref{eq:clear} is monotone in the fields, which gives the inclusion. A measurement named by $\sigma_i$ was reported by some $k\in\mathcal{K}_i$ from a grid $x\in\mathcal{X}_i$ on a band $f\in\mathcal{F}_i$ at a round $t\in\mathcal{T}_i$, so its block lies in $\mathcal{E}(\sigma_i)$ and is cleared. Eq.~\eqref{eq:map} depends on the ledger only through the sums and counts of its entries, and clearing removes both.\end{proof}

According to \textbf{Proposition}~\ref{prop:complete}, a merged execution cannot remove less than the demands it replaces, so a record issued for the campaign of $\sigma_{\mathcal{G}}$ holds for each of its members.

\begin{proposition}\label{prop:interval}
Suppose that all demands share their terminals, grids, and bands, that their intervals start and end at block boundaries, and that no interval contains another. Then $(\mathrm{P})$ has an optimal solution whose groups are runs of demands that are consecutive in start round, and Eq.~\eqref{eq:dp} attains its optimum $F(Q)$ with $O(Q^{2})$ evaluations of $C$ and $\Delta$.
\end{proposition}

\begin{proof}
Take an optimal cover and replace each group $\mathcal{G}$ by the run of all demands whose intervals lie in the interval of $\sigma_{\mathcal{G}}$. The envelope and the cost are unchanged, and the right side of $(\mathrm{C}_1)$ can only grow, so the run is feasible. Since no interval contains another, sorting by start round also sorts by end. The blocks of the envelope of a run are then those of its member intervals plus those of the gaps between consecutive members, and the over-deletion is the utility of the gap blocks that hold a named report. If two runs overlap, shrink the later one so that it starts after the earlier ends. The shorter run has a smaller envelope and a subset of the gaps, so its cost and its over-deletion do not grow. The result is a no costlier partition into runs, and Eq.~\eqref{eq:dp} examines each such partition by its last run, giving the optimum and the count.\end{proof}

\begin{theorem}\label{thm:greedy}
Let $\mathrm{OPT}_{\mathcal{A}}$ be the least cost of covering $\mathcal{Q}$ with candidates in $\mathcal{A}$. The greedy of Section~\ref{subsec:merge} returns a cover meeting $(\mathrm{C}_1)$ that costs at most $(1+\ln Q)\,\mathrm{OPT}_{\mathcal{A}}$. Under the conditions of \textbf{Proposition}~\ref{prop:interval}, $\mathrm{OPT}_{\mathcal{A}}$ is the optimum of $(\mathrm{P})$.
\end{theorem}

\begin{proof}
Every candidate in $\mathcal{A}$ satisfies $(\mathrm{C}_1)$, and the greedy stops only when every demand is covered. With the candidates as sets and $C(\sigma_{\mathcal{G}})$ as their weights, the bound is the analysis of greedy weighted set cover~\cite{Chvatal1979MOR}. Under the conditions of \textbf{Proposition}~\ref{prop:interval}, the run from demand $i$ to demand $j$ is exactly $\mathcal{G}_{ij}$, so $\mathcal{A}$ contains the optimal runs.\end{proof}

According to \textbf{Theorem}~\ref{thm:greedy}, the greedy loses at most a logarithmic factor to the best cover over its candidates, and Section~\ref{sec:exp} compares it with exhaustive search.

\begin{proposition}\label{prop:flush}
Let demands arrive as a Poisson process with rate $\alpha$, let their liability rates $\psi_i$ be independent and identically distributed, and let each flush cost $c^{0}$ plus a term independent of when it happens. Among rules that flush based on the queue, the one minimizing the long-run average of flush cost plus liability flushes once $\Psi\ge\theta$ for some threshold $\theta$.
\end{proposition}

\begin{proof}
Between arrivals, $\Psi$ is constant and waiting only adds liability, so flushing is decided at arrivals, with state $\Psi$. In the average-cost optimality equation, flushing costs $c^{0}$ and restarts from $\Psi=0$, independent of $\Psi$. Waiting costs $\Psi/\alpha$ until the next arrival and moves the state to $\Psi+\psi_i$. Value iteration preserves monotonicity in $\Psi$, so the cost of waiting grows with $\Psi$ and crosses the constant flushing cost once, at $\theta$.\end{proof}

According to \textbf{Proposition}~\ref{prop:flush}, the flush rule needs only $\Psi$. Merging makes a flush cheaper per demand and favors a larger $\theta$, which Section~\ref{sec:exp} tunes on a validation trace.

\begin{figure*}[!t]
\centering
\subfloat[Over-deletion bound.]{\includegraphics[width=0.24\textwidth]{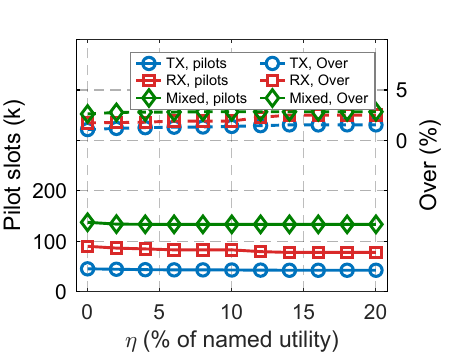}\label{fig:eta}}
\hfill
\subfloat[Arrival rate.]{\includegraphics[width=0.24\textwidth]{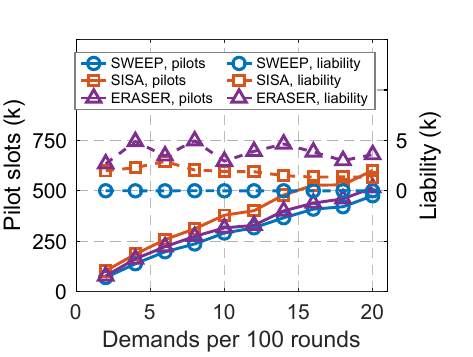}\label{fig:rate}}
\hfill
\subfloat[Block length.]{\includegraphics[width=0.24\textwidth]{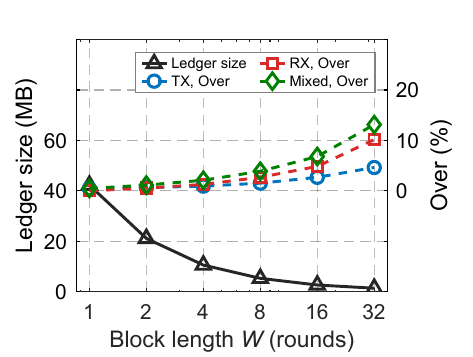}\label{fig:block}}
\hfill
\subfloat[Fault overlap.]{\includegraphics[width=0.24\textwidth]{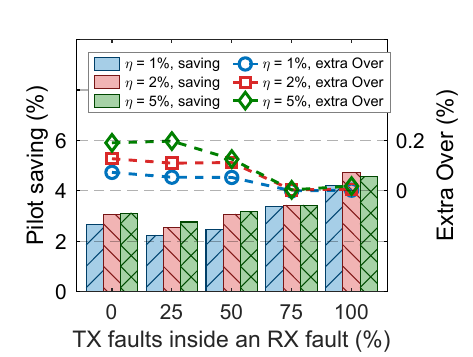}\label{fig:overlap}}
\caption{Results of \method{}. (a) Pilot slots and over-deletion as $\eta$ varies, where $\eta=0$ merges only nested demands. (b) Liability and pilot slots as the arrival rate of demands grows. (c) Ledger size and over-deletion as the block length $W$ varies. (d) Pilot saving and added over-deletion of merging against one-by-one execution as the overlap between transmitter and receiver faults grows.}
\label{fig:results}
\end{figure*}

\section{Performance Evaluation}
\label{sec:exp}

\subsection{Experimental Setup}
\label{subsec:setup}

We simulate a BS at the edge of a $100\times100$ m urban area of $20\times20$ grids, with channels ray-traced on a street layout at 2.4 and 5.8 GHz. The BS has $M=16$ antennas and $B=32$ beams, and the map offers $L=4$ candidates. $K=50$ terminals move by random waypoints at walking speed and report every round for 400 rounds, with $W=5$ and $n^{\textrm{min}}=3$. The traffic $\lambda_{x,f}$ follows the visits to each grid. A power fault adds 10 dB, a positioning fault files reports four to six grids away, and a receive-chain fault adds $180^{\circ}$ to eight of the 16 chains at 5.8 GHz. Faults last 30 to 60 rounds, and their demands arrive as a Poisson process 5 to 20 rounds later. By default, $c^{0}$ equals 50 full sweeps and $\eta$ is 2\% of the named utility. Thresholds are tuned on a separate trace, and Figs.~\ref{fig:eta} and~\ref{fig:overlap} flush at $\theta=50$. On the testbed in Fig.~\ref{fig:testbed}, a USRP-2955 with four coherent receive ports serves as the BS and forms $B=8$ beams. Three USRP-2900 terminals are carried by people along fixed routes over 12 grids of 1 m in a 55 m$^{2}$ lab at 5.8 GHz. The faults are injected by a 6 dB gain, a wrong grid, and a phase offset on one receive channel.

We measure the map by the traffic-weighted SNR loss of Eq.~\eqref{eq:loss} and the hit rate, the share of selections whose best beam is among the $L$ candidates, both averaged after round 50. $\mathrm{GainErr}$ is the gap in dB between the map gain and the measured gain of the selected beam, and $\mathrm{Invalid}$ is the traffic share of pairs that are not valid. $\mathrm{Left}$ is the share of the named measurements filed so far that remain, and $\mathrm{Over}$ is the share of good measurements cleared. $\mathrm{Wait}$ counts the rounds from the arrival of a demand until the pairs carrying 90\% of its traffic hold $n^{\textrm{min}}$ new measurements. The cost counts the campaigns, the pilot slots of campaigns and of full sweeps on invalid pairs, and the liability, the sum of $\psi_i$ over the waiting rounds. We compare \method{} with five baselines. SISA~\cite{Bourtoule2021SP} splits the grids into four shards and, whenever the queue holds a fixed number of demands, runs one campaign per shard on the envelope of the demands with blocks in it. ERASER~\cite{Hu2024CCS} defers demands until the recent SNR loss exceeds its earlier level by a threshold, then executes the envelope of the queue. FedRecover~\cite{Cao2023SP} removes every measurement of the named terminals, and Crab~\cite{Jiang2025TIFS} rolls the map back to before the fault and re-sounds every pair, both on arrival. ST-REM~\cite{Hu2020TON} screens reports on arrival against trusted measurements and never deletes. Retrain rebuilds the map from the raw reports logged by the simulator and testbed, except the named ones.

\subsection{Experimental Results}
\label{subsec:results}

\begin{figure}[t]
\centering
\subfloat[Queue length.]{\includegraphics[width=0.485\linewidth]{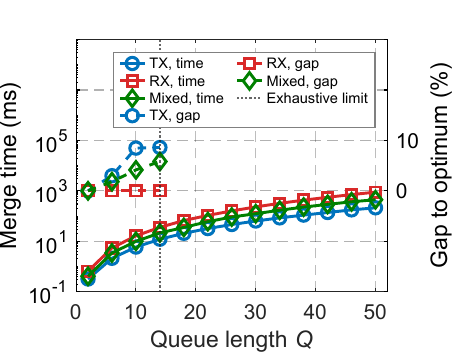}\label{fig:scale}}
\hfill
\subfloat[Number of grids.]{\includegraphics[width=0.485\linewidth]{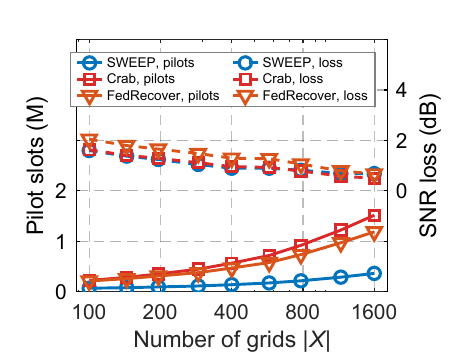}\label{fig:grid}}
\caption{Scalability of \method{}. (a) Merge time and gap to exhaustive search as $Q$ grows. (b) Pilot slots and SNR loss as the number of grids grows.}
\label{fig:scalability}
\end{figure}

In Table~\ref{tab:main}, \method{} clears the fewest good measurements in every fault group, at most 2.6\% in simulation and 5.1\% on the testbed, while FedRecover and Crab clear up to 86.8\%. Under mixed faults, \method{} also spends the fewest pilot slots, with liability 3 against 4889 for ERASER. Its SNR loss is the closest to Retrain under receiver and mixed faults. SISA spends fewer pilot slots under receiver faults but clears 13.0\%, and ST-REM deletes nothing after screening on arrival.

In Fig.~\ref{fig:results}, raising $\eta$ to 20\% saves up to 13.6\% of the pilot slots for less than 0.8 points of over-deletion. \method{} spends 20.6\% to 32.1\% fewer pilot slots than SISA at every arrival rate. Each doubling of $W$ halves the ledger and raises the over-deletion, and merging saves up to 4.7\% of the pilot slots. In Fig.~\ref{fig:scalability}, the greedy stays within 8.6\% of exhaustive search, and Crab spends up to 4.3 times the pilot slots of \method{}.

\section{Conclusion}
\label{sec:conclusion}

We studied how a base station forgets faulty measurements in a crowdsourced channel knowledge map. SDUs remove transmitter and receiver faults without discarding whole terminals or rolling back the map, and \method{} merges them into re-sounding campaigns whose over-deletion stays within a bound.
In simulation and on a USRP testbed, \method{} clears at most 5.1\% of the good measurements, against up to 86.8\% for removing whole terminals or rolling back the map.

\bibliographystyle{IEEEtran}
\bibliography{ref}

@article{Romandini2025TNNLS,
  author={Romandini, Nicolò and Mora, Alessio and Mazzocca, Carlo and Montanari, Rebecca and Bellavista, Paolo},
  journal={IEEE Transactions on Neural Networks and Learning Systems},
  title={{Federated Unlearning: A Survey on Methods, Design Guidelines, and Evaluation Metrics}},
  year={2025},
  volume={36},
  number={7},
  pages={11697-11717},
  doi={10.1109/TNNLS.2024.3478334}
}

@article{Wu2026TNSE,
  author={Wu, Beining and Ding, Zihao and Huang, Jun},
  journal={IEEE Transactions on Network Science and Engineering},
  title={{A Review of Continual Learning in Edge AI}},
  year={2026},
  volume={13},
  pages={6571--6588},
  doi={10.1109/TNSE.2026.3657652}
}

@inproceedings{Bourtoule2021SP,
  author={Bourtoule, Lucas and Chandrasekaran, Varun and Choquette-Choo, Christopher A. and Jia, Hengrui and Travers, Adelin and Zhang, Baiwu and Lie, David and Papernot, Nicolas},
  booktitle={Proceedings of the IEEE Symposium on Security and Privacy (S\&P)},
  title={{Machine Unlearning}},
  year={2021},
  pages={141--159}
}

@article{Zeng2024COMST,
  author={Zeng, Yong and Chen, Junting and Xu, Jie and Wu, Di and Xu, Xiaoli and Jin, Shi and Gao, Xiqi and Gesbert, David and Cui, Shuguang and Zhang, Rui},
  journal={IEEE Communications Surveys \& Tutorials},
  title={{A Tutorial on Environment-Aware Communications via Channel Knowledge Map for 6G}},
  year={2024},
  volume={26},
  number={3},
  pages={1478--1519},
  doi={10.1109/COMST.2024.3364508}
}

@inproceedings{Wu2021ICCW,
  author={Wu, Di and Zeng, Yong and Jin, Shi and Zhang, Rui},
  booktitle={Proceedings of the IEEE International Conference on Communications Workshops (ICC Workshops)},
  title={{Environment-Aware and Training-Free Beam Alignment for mmWave Massive MIMO via Channel Knowledge Map}},
  year={2021},
  pages={1--7},
  doi={10.1109/ICCWorkshops50388.2021.9473871}
}

@inproceedings{Ye2017LCN,
  author={Ye, Yanzhen and Wang, Bang},
  booktitle={Proceedings of the IEEE 42nd Conference on Local Computer Networks (LCN)},
  title={{On Building an Indoor Radio Map from Crowdsourced Samples with Annotation Errors}},
  year={2017},
  pages={401--408},
  doi={10.1109/LCN.2017.26}
}

@inproceedings{Gong2025ICCCW,
  author={Gong, Yuanxiang and Wang, Xingchen and Wang, Ting and Zhang, Chiya},
  booktitle={Proceedings of the IEEE/CIC International Conference on Communications in China Workshops (ICCC Workshops)},
  title={{Channel Knowledge Map Updating with Machine Unlearning}},
  year={2025},
  pages={1--6},
  doi={10.1109/ICCCWorkshops67136.2025.11148163}
}

@inproceedings{Hu2024CCS,
  author={Hu, Yuke and Lou, Jian and Liu, Jiaqi and Ni, Wangze and Lin, Feng and Qin, Zhan and Ren, Kui},
  booktitle={Proceedings of the ACM SIGSAC Conference on Computer and Communications Security (CCS)},
  title={{ERASER: Machine Unlearning in MLaaS via an Inference Serving-Aware Approach}},
  year={2024},
  pages={3883--3897},
  doi={10.1145/3658644.3670398}
}

@inproceedings{Cao2023SP,
  author={Cao, Xiaoyu and Jia, Jinyuan and Zhang, Zaixi and Gong, Neil Zhenqiang},
  booktitle={Proceedings of the IEEE Symposium on Security and Privacy (S\&P)},
  title={{FedRecover: Recovering from Poisoning Attacks in Federated Learning using Historical Information}},
  year={2023},
  pages={1366--1383},
  doi={10.1109/SP46215.2023.10179336}
}

@article{Jiang2025TIFS,
  author={Jiang, Yu and Shen, Jiyuan and Liu, Ziyao and Tan, Chee Wei and Lam, Kwok-Yan},
  journal={IEEE Transactions on Information Forensics and Security},
  title={{Toward Efficient and Certified Recovery From Poisoning Attacks in Federated Learning}},
  year={2025},
  volume={20},
  pages={2632--2647},
  doi={10.1109/TIFS.2025.3533907}
}

@article{Hu2020TON,
  author={Hu, Yidan and Zhang, Rui},
  journal={IEEE/ACM Transactions on Networking},
  title={{A Spatiotemporal Approach for Secure Crowdsourced Radio Environment Map Construction}},
  year={2020},
  volume={28},
  number={4},
  pages={1790--1803},
  doi={10.1109/TNET.2020.2992939}
}

@article{Chvatal1979MOR,
  author={Chvatal, V.},
  journal={Mathematics of Operations Research},
  title={{A Greedy Heuristic for the Set-Covering Problem}},
  year={1979},
  volume={4},
  number={3},
  pages={233--235},
  doi={10.1287/moor.4.3.233}
}

@article{Tao2024VLDB,
  author={Tao, Youming and Wang, Cheng-Long and Pan, Miao and Yu, Dongxiao and Cheng, Xiuzhen and Wang, Di},
  journal={Proceedings of the VLDB Endowment},
  title={{Communication Efficient and Provable Federated Unlearning}},
  year={2024},
  volume={17},
  number={5},
  pages={1119--1131},
  doi={10.14778/3641204.3641220}
}

@inproceedings{Hu2019INFOCOM,
  author={Hu, Yidan and Zhang, Rui},
  booktitle={Proceedings of the IEEE Conference on Computer Communications (INFOCOM)},
  title={{Differentially-Private Incentive Mechanism for Crowdsourced Radio Environment Map Construction}},
  year={2019},
  pages={1594--1602},
  doi={10.1109/INFOCOM.2019.8737512}
}

@inproceedings{Han2025ICCCW,
  author={Han, Rubing and Zhang, Chiya and Wang, Ting},
  booktitle={Proceedings of the IEEE/CIC International Conference on Communications in China Workshops (ICCC Workshops)},
  title={{Analysis of Channel Knowledge Map Updating from Information Perspective}},
  year={2025},
  pages={1--6},
  doi={10.1109/ICCCWorkshops67136.2025.11148189}
}

@inproceedings{Liu2021IWQoS,
  author={Liu, Gaoyang and Ma, Xiaoqiang and Yang, Yang and Wang, Chen and Liu, Jiangchuan},
  booktitle={Proceedings of the IEEE/ACM 29th International Symposium on Quality of Service (IWQoS)},
  title={{FedEraser: Enabling Efficient Client-Level Data Removal from Federated Learning Models}},
  year={2021},
  pages={1--10},
  doi={10.1109/IWQOS52092.2021.9521274}
}

\end{document}